\documentclass[11pt]{article}
\usepackage[latin1]{inputenc}
\usepackage{fullpage}
\usepackage{graphicx}
\usepackage{epsfig}
\usepackage{color}
\usepackage{lscape}
\usepackage{verbatim}
\usepackage{amsthm, amssymb}
\usepackage{amsmath}
\newtheorem{Lem}{Lemma}
\newtheorem{theorem}{Theorem}
\newtheorem{Cor}{Corollary}
\newtheorem{Obs}{Observation}

\newcommand{\lca}{\operatorname{lca}}

\newcommand{\rep}{\operatorname{rep}}
\newcommand{\ang}[1]{\left< #1 \right>}
\newcommand{\rt}{\operatorname{root}}
\newcommand{\Patrascu}{P\v{a}tra\c{s}cu}
\newcommand{\poly}{\operatorname{poly}}

\newcommand{\pre}{\operatorname{pre}}
\newcommand{\post}{\operatorname{post}}
\newcommand{\sort}{\operatorname{sort}}
\newcommand{\pred}{\operatorname{pred}}

\title{Connectivity Oracles for Planar Graphs} \author{Glencora
  Borradaile\footnote{School of Electrical Engineering and Computer
    Science, Oregon State University, {\tt glencora@eecs.orst.edu}.
    Supported by NSF CCF-0963921.} \and Seth Pettie\footnote{Supported
    by NSF CAREER grant no. CCF-0746673 and a grant from the US-Israel
    Binational Science Foundation.} \and Christian Wulff-Nilsen
  \footnote{Department of Mathematics and Computer Science, University of Southern Denmark, \texttt{koolooz@diku.dk}, \texttt{http://imada.sdu.dk/$_{\widetilde{~}}$cwn/}.
    This research was partially supported by NSERC and MRI.}}

\begin{document}

\maketitle
\begin{abstract}
We consider dynamic subgraph connectivity problems for planar undirected graphs.
In this model there is a fixed underlying planar graph, where each
edge and vertex is either ``off'' (failed) or ``on'' (recovered).  We
wish to answer connectivity queries with respect to the ``on''
subgraph.
The model has two natural variants, one in which there are $d$
edge/vertex failures that precede all connectivity queries, and one in
which failures/recoveries and queries are intermixed.

We present a $d$-failure connectivity oracle for planar graphs that
processes any $d$ edge/vertex failures in $\sort(d,n)$ time so that
connectivity queries can be answered in $\pred(d,n)$ time.  (Here
$\sort$ and $\pred$ are the time for integer sorting and integer
predecessor search over a subset of $[n]$ of size $d$.)  Our algorithm
has two discrete parts.  The first is an algorithm tailored to
triconnected planar graphs.  It makes use of Barnette's theorem, which
states that every triconnected planar graph contains a degree-3
spanning tree.
The second part is a generic reduction from general (planar) graphs to
triconnected (planar) graphs.  
Our algorithm is, moreover, provably optimal.  An implication
of \Patrascu{} and Thorup's lower bound on predecessor search is 
that no $d$-failure connectivity oracle (even on trees) can beat $\pred(d,n)$ query time.

We extend our algorithms to the
subgraph connectivity model where edge/vertex failures (but no
recoveries) are intermixed with connectivity queries.  In triconnected
planar graphs each failure and query is handled in $O(\log n)$ time
(amortized), whereas in general planar graphs both bounds become
$O(\log^2 n)$.

\end{abstract}
\setcounter{page}{0}
\thispagestyle{empty}

\newpage

\section{Introduction}

Algorithms for dynamic graphs have traditionally assumed that the graph evolves according to a completely arbitrary sequence of insertions and deletions of graph elements.  This model makes minimal assumptions but often sacrifices efficiency for generality.  For example, real world networks (router networks, road networks, etc.) {\em do} change slowly over time.  However, the real dynamism of the networks comes from the frequent {\em failure} of edges/nodes and their subsequent recovery.  In this paper we study connectivity problems in the {\em dynamic subgraph model}, which attempts to accurately model this type of dynamism.  It is assumed that there is some fixed underlying graph whose nodes and edges can be {\em off} (failed) or {\em on}; queries (connectivity queries, in our case) then answer questions about the subgraph turned on.
The power of this model (compared to the fully dynamic graph model) stems from the ability to preprocess the underlying graph in advance.

There are two natural variants of the dynamic subgraph model.  In the $d$-failure version failures and recoveries occur in lockstep: a set $F$ of $d=|F|$ edges/nodes fail together.
Our goal is to process $F$, ideally in $\tilde{O}(d)$ time, in order to answer connectivity queries in $G\backslash F$.  Here $d$ may or may not be a parameter of the algorithm.  In the fully dynamic model, node/edge failures and recoveries are presented one at a time and intermixed with connectivity queries, whereas in the decremental model the updates are restricted to failures.

\paragraph{Results.} We give new algorithms for subgraph connectivity (aka connectivity oracles) on undirected planar graphs in the $d$-failure model and the decremental model, all of which require linear preprocessing time.  When failures are restricted to edges, we give an especially simple connectivity oracle that processes $d$ edge failures (for any $d$) in $O(\sort(d,n))$ time and subsequently answers queries in $O(\pred(d,n))$ time.
Here $\sort$ and $\pred$ refer to the time for sorting $d$ integers in the universe $\{1,\ldots,n\}$ and $\pred$ for the time for predecessor search, given $\sort(d,n)$ preprocessing time.  (It is known that $\sort(d,n)=O(d \log\log d)$ deterministically, $O(d \sqrt{\log\log d})$ randomized~\cite{Han04,HT02}, and $O(d)$ randomized if $d < 2^{\sqrt{\log n}-\epsilon}$~\cite{AnderssonHNR98}.   For predecessor search the bound is $\pred(d,n) = O(\min\{\frac{\log d}{\log\log n},\frac{\log\log d \cdot \log\log n}{\log\log\log n}\})$ deterministically~\cite{FW94,AnderssonT07} and $O(\log\log(n/d))$ randomized~\cite{vEKZ77,Willard83}.)

The problem becomes more complicated when vertices fail since we cannot, in general, spend time proportional to their degrees.  Our second algorithm is a $d$-failure connectivity oracle for edge and vertex failures with the same parameters (linear preprocessing, $O(\sort(d,n))$ to process $d$ failures, $O(\pred(d,n))$ time per query).  It consists of two parts: a solution for triconnected graphs and a generic reduction from $d$-failure oracles in general graphs to $d$-failure oracles in triconnected graphs.  Triconnectivity plays an important role in the algorithm as it allows us to apply Barnette's theorem~\cite{Barnette66}, which states that every triconnected planar graph contains a degree-3 spanning tree.
It is known~\cite{PatrascuT07} that predecessor search is reducible to the $d$-edge/node failure connectivity problem on trees (and hence planar graphs).  
Our query time is therefore provably optimal.  In particular, \Patrascu{} and Thorup's lower bound~\cite{PatrascuT06} implies that $O(\log\log(n/d))$ query time cannot be beaten in general,
even given $O(d\poly(\log n))$ time to preprocess the edge/node failures.

Our third algorithm is in the decremental model.  In triconnected planar graphs we can support vertex failures in $O(\log n)$ amortized time and connectivity queries in $O(\log n)$ time, whereas in general planar graphs both bounds become $O(\log^2 n)$.\footnote{These bounds are amortized over the actual number $f$ of failures, i.e., in triconnected graphs processing $f$ failures takes $O(f\log n)$ time total.}  The logarithmic slowdown comes from a new reduction from (planar) dynamic subgraph connectivity to the same problem in triconnected graphs.

\paragraph{Prior Work on Dynamic Connectivity.} Before surveying
dynamic subgraph connectivity it is instructive to see what type of
``off the shelf'' solutions are available using traditional dynamic
graph algorithms.  The best known dynamic connectivity algorithms for
general undirected graphs require $O(\sqrt{n})$ worst case time per
edge update~\cite{Frederickson85,EppsteinGIN97} or $O(\log^2 n)$ time
amortized~\cite{HolmLT01}.  Vertex updates are simulated with $O(n)$
edge updates.  In dynamic planar graphs the best connectivity
algorithms take $O(\log n)$ time per edge update~\cite{EppsteinITTWY92}.

\paragraph{Prior Work on Subgraph Connectivity.}  The dynamic subgraph model was proposed explicitly by Frigioni and Italiano~\cite{FrigioniI00}, who proved that in planar graphs, node failures and recoveries could be supported in $O(\log^3 n)$ amortized time per operation and connectivity queries in $O(\log^3 n)$ worst case time.  An earlier algorithm of Eppstein et al.~\cite{EppsteinITTWY92} implies that edge failures, edge recoveries, and connectivity queries in planar graphs require $O(\log n)$ time.  In general graphs, Chan, \Patrascu, and Roditty~\cite{ChanPR11} (improving~\cite{Chan06b}) showed that node updates could be supported in amortized time $O(m^{2/3})$, where $m$ is the number of edges,
and connectivity queries in $O(m^{1/3})$.  Chan et al.~\cite{AfshaniC09,Chan06b,ChanPR11} gave numerous applications of subgraph connectivity to geometric connectivity problems.  The first algorithm with worst-case guarantees was given by Duan~\cite{Duan10}, who showed that node updates and queries require only $O(m^{4/5})$ and $O(m^{1/5})$ time, respectively.  

In the $d$-{\em edge} failure model \Patrascu{} and Thorup~\cite{PatrascuT07} gave a connectivity oracle for general graphs that processes $d$ failures in $O(d\log^2 n\log\log n)$ time and answers queries in $O(\log\log n)$ time.  However, their structure requires {\em exponential} preprocessing time; a variant constructible in polynomial time has a slower update time: $O(d\log^{5/2} \log\log n)$.  Duan and Pettie~\cite{DuanP10} gave a connectivity oracle in the $d$-{\em node} failure model occupying $O(mn^{\epsilon})$ space that processes failures in $O(\poly(d,\log n))$ time and answers queries in $O(d)$ time, where $\poly(d,n)$ depends on $\epsilon$.  They also showed that a $d$-edge failure oracle could be constructed in $\tilde{O}(n)$ time with $O(d^2\log\log n)$ update time and $O(\log\log n)$ query time.  

The {\em distance sensitivity} oracles avoiding 1 node failure~\cite{DTCR08,BernsteinK09} or 2 node failures \cite{DP09a} also, as a special case, answer 1- and 2-failure connectivity queries on {\em directed} graphs in $O(1)$ and $O(\log n)$ time, respectively.
These data structures occupy $\tilde{O}(n^2)$ space.

\paragraph{Overview}
Section~\ref{sec:DefsNotRes} reviews notation and terminology.  In Section~\ref{sec:MultiEdge} we describe our planar $d$-edge failure connectivity oracle.  In Section~\ref{sec:MultiVertexEdge} we give a $d$-vertex failure oracle for triconnected planar graphs
and in Section~\ref{sec:SingleVertexEdge} we extend it to a decremental subgraph connectivity oracle for triconnected graphs.
Finally, in Section~\ref{sect:reduction-to-tri} we give reductions from general graphs to triconnected graphs, which do not assume planarity.

\section{Definitions, Notation, and Basic Results}\label{sec:DefsNotRes}
We assume that all graphs considered are undirected.
A \emph{planar graph} is a graph that can be drawn in the plane such that no two edges cross. We refer to such a drawing
as a \emph{plane graph}. A plane graph $G = (V,E)$ partitions the plane into maximal open connected sets and we refer to the closure
of these sets as the \emph{faces} of $G$. If $G$ is connected, we define a plane graph, called the \emph{dual graph} $G^*$ of $G$,
as follows. Associated with each face $f$ of $G$ is a vertex in $G^*$ which we identify with $f$ and which we draw inside $f$.
For each edge $e$ in $G$, there is an edge $(f_1,f_2)$ in $G^*$, where $f_1$ and $f_2$ are the faces in $G$ incident to $e$.
We draw $(f_1,f_2)$ such that it crosses $e$ exactly once and crosses
no other edge in $G$ or in $G^*$. We identify $(f_1,f_2)$ with $e$. It can be shown that $G^*$ is also connected and
that $(G^*)^* = G$. In particular, each face in $G^*$ corresponds to a vertex in $G$. We refer to $G$ as
the \emph{primal} graph.

For vertices $u$ and $v$ in a rooted tree $T$, we let $\lca_T(u,v)$ denote the lowest common ancestor of $u$ and $v$ in $T$.
Consider a spanning tree $T$ in primal graph $G$. It is well-known that the edges not present in $T$ form a spanning tree $T^*$ of
$G^*$. We call $T^*$ the \emph{dual tree} of $T$. The following lemma is a well-known result~\cite{Whitney1933}.
\begin{Lem}\label{Lem:PrimalDualLink}
  For an edge $e$ in primal tree $T$, let $f$ and $g$ be the faces to
  either side of $e$ in $G$. Then the edges of $E\setminus\{e\}$ that
  connect the two subtrees of $T\setminus\{e\}$ are exactly those on
  the simple path in $T^*$ from $f$ to $g$.
\end{Lem}

A co-path is a sequence of faces that is a sequence of vertices in the
dual that form a path.  A co-path {\em avoids} a set of edges $F$ if
every consecutive pair of faces in the sequence shares an edge not in $F$.

\section{Edge Failures}\label{sec:MultiEdge}

In this section, we develop a data structure that, given a planar
undirected graph $G = (V,E)$, an integer $d\geq 1$, and a dynamic
subset of at most $d$ \emph{failed edges}, supports the following
operations:
\begin{enumerate}
\item \texttt{update}$(F)$: set $F$ to be the set of failed edges;
\item \texttt{connected?}$(u,v)$: are vertices $u$ and $v$
  connected in $G\setminus F$?
\end{enumerate}
We may assume that $G$ is plane and connected.  We will examine this
problem in the dual by way of the following lemmas:
\begin{Lem}\label{Lem:FacePath}
  Two vertices $u$ and $v$ are connected in $G\setminus F$ iff there
  is a $u$-to-$v$ co-path in $G^*$ avoiding $F$.
\end{Lem}
\begin{proof}
  If $u$ and $v$ are connected in $G\setminus F$, then there is a
  $u$-to-$v$ path $P$ in $G$ not using any edges in $F$.  This path,
  viewed as a sequence of vertices in $G$ is a sequence of faces, or
  co-path, in $G^*$.  Since consecutive vertices in $P$ are adjacent by
  way of edges not in $F$, consecutive faces in the identified co-path
  share an edge not in $F$: this $u$-to-$v$ co-path avoids $F$.

  Conversely, let $P$ be a $u$-to-$v$ co-path avoiding $F$.  $P$ is a
  sequence of faces in $G^*$, and so is a sequence of vertices in $G$.
  Consecutive faces on $P$ in $G^*$ share an edge avoiding $F$ and so
  are adjacent by way of edges not in $F$: $P$ is a path in $G$ not
  using any edges in $F$.
\end{proof}

Consider the subgraph $G_F^*$ of $G^*$ consisting of the failed edges,
$F$. Let $G_F^*$ inherit the embedding of $G^*$.  We refer to the
faces of $G_F^*$ as \emph{superfaces}.  Each superface corresponds to
the union of faces of $G^*$.
\begin{Lem}\label{Lem:SuperFaces}
  Let $f$ and $g$ be faces of $G^*$.  There is an $f$-to-$g$ co-path
  in $G^*$ avoiding $F$ iff $f$ and $g$ are contained in the same
  superface of $G_F^*$.
\end{Lem}
\begin{proof}
  Suppose $f$ and $g$ are contained in different superfaces; let $C$
  be the set of edges that bound the superface containing $f$.  Any
  $f$-to-$g$ co-path requires a face on either side of $C$ and so
  cannot avoid $C$, which is a subset of $F$.

  Suppose otherwise; let $S$ be the set of faces of $G^*$ that form
  the superface $f_S$ containing $f$ and $g$.  Since $f_S$ is a face of
  $G_F^*$, there is a curve $\cal C$ contained entirely in $f_S$ that starts
  inside $f$ and ends inside $g$.  Consider the sequence of faces in
  $S$ that this path visits;  this is an $F$-avoiding $f$-to-$g$ co-path as
  consecutive faces must share an edge that $\cal C$ crosses and this
  edge cannot be in $F$.
\end{proof}

We restate the operations in light of the corollary:
\begin{Cor} \label{cor:connected}
  Two vertices $u$ and $v$ are connected in $G\setminus F$ iff $u$ and
  $v$ are contained in the same superface of $G_F^*$.
\end{Cor}

\begin{enumerate}
\item \texttt{update}$(F)$: set $F$ to be the set of failed edges;
\item \texttt{connected?}$(u,v)$: are $u$ and $v$ contained in the
  same superface of $G_F^*$?
\end{enumerate}

\subsection{The data structure}
Fix a rooted spanning tree $T$ of the primal graph $G$.  We use $T$ to
determine the superfaces of $G_F^*$ containing the faces corresponding to the query vertices $u$ and $v$.
To do so, we require constant-time lca query support for $T$~\cite{HT84,AlstrupGKR02,BF-C00}.

\paragraph{\texttt{update}$(F)$:} The edges $F$ are listed in
no particular order.  We start by building the planar embedding of the
subgraph $G_F^*$ induced by $F$ that is inherited from $G^*$.  Let
$V(F)$ be the endpoints of $F$ (in the dual sense).  For each vertex
$v \in V(F)$ identify the edges of $F$ incident to $v$ and their
cyclic ordering\footnote{The cyclic ordering of edges incident to each
  vertex is sufficient to define the embedding.~\cite{Edmonds60}} in $G^*$.
We can compute these orderings in $\sort(d,n)$ time.

The boundaries of the superfaces given by $G_F^*$ can be traversed in
$O(d)$ time given this combinatorial embedding.  The set $V(F)$ is the
set of faces of $G^*$ that are along the boundaries of superfaces of
$G_F^*$.  Label a dual face/primal vertex in $V(F)$ with the
superface that contains it; mark the vertices of the static tree $T$
with these superface labels.

\paragraph{\texttt{connected?}$(u,v)$:}
To answer a query \texttt{connected?}$(u,v)$, suppose we have
identified the first and last marked vertices (if any) on the
$u$-to-$v$ path in $T$; call them $\hat{u}$ and $\hat{v}$.\footnote{This is a slightly more general problem than the {\em marked ancestor} problem considered by Alstrup, Husfeldt, and Rauhe~\cite{AlstrupHR98}.}
By Lemmas~\ref{Lem:FacePath}
and~\ref{Lem:SuperFaces}, $u$ and $v$ are connected in $G\setminus F$
iff $\hat{u}$ and $\hat{v}$ do not exist, or $\hat{u}$ and $\hat{v}$ are labelled with the same superface.  
Therefore, we need only identify $\hat{u}$ and $\hat{v}$, if they exist.  
Lemma~\ref{Lem:marked-ancestor} shows that finding $\hat{u}$ and $\hat{v}$ is reducible to 
one least common ancestor query and $O(1)$ predecessor queries.

\begin{Lem}\label{Lem:marked-ancestor}
Let $T$ be a tree of size $n$ and $M\subseteq V(T)$ be a set of {\em marked} vertices, with $|M|=d$.
Then after $O(n)$ preprocessing (independent of $M$),
an $O(d)$-size data structure can be constructed in $O(\sort(d,n))$ time that answers the
following query in $O(\pred(d,n))$ time:
Given $u,v\in V(T)$, what are the first and last $M$-vertices on the $u$-$v$ path?
\end{Lem}

\begin{proof}
Recall that $T$ is rooted.  Let $w = \lca(u,v)$.  
The first marked vertex on the $u$-to-$v$ path is either the first marked vertex on the $u$-to-$w$
path or the last marked vertex on the $v$-to-$w$ path.  
Thus, we can assume without loss of generality that $v$ is an ancestor of $u$.
Fix an arbitrary DFS of $T$ and let $\pre(x)$ (resp.~$\post(x)$) be the
time when $x$ is pushed onto (resp.~popped off) the stack during DFS.
Given $M$, we first sort the $\pre$ and $\post$ indices of its elements, in $O(\sort(d,n))$ time, which allows us to label each $x\in M$ with the nearest strictly ancestral $M$-vertex $\mu(x)$.  Here it is convenient to assume that the root is marked (honorarily, if not in $M$) so $\mu(x)$ is defined everywhere except the root.  We build a global predecessor structure on the set $\mathcal{S} = \{\pre(x), \post(x) \;:\; x\in M\}$ and local predecessor structures on the sets $\mathcal{S}_x = \{\pre(x') \;:\; x'\in \mu^{-1}(x)\}$, where $x\in M$ and $\mathcal{S}_x$ is the set of ``immediate'' descendants in $M$ connected by a path of non-$M$ vertices.
To answer a query $u,v$ (where $v$ is ancestral to $u$) we first find the closest marked ancestors $u',v'\in M$ as follows.  
Let $i$ be the predecessor of $\pre(u)$ in $\mathcal{S}$.  If $i = \pre(y)$ for some $y\in M$ then $x$ is an immediate descendant of $y$ and $u'=y$.  If $i = \post(y)$ then let $u' = \mu(y)$.  See Figure~\ref{fig:marked-ancestor}(a).  It follows that $u'$ is an ancestor of $u$ and that there are no other $M$-vertices on the path from $u$ to $u'$.  If $u'$ is ancestral to $v$ then there are no marked nodes on the $u$-$v$ path, so assume this is not the case.
In order to find the {\em last} marked vertex on the path from $u$ to $v$ we search for the predecessor of $\pre(u)$ in $\mathcal{S}_{v'}$, say $\pre(y)$.  Since there is {\em some} marked vertex on the path from $u$ to $v'$ it follows that $u$ is a descendant of $y$, which is a descendant of $v$ and that there are no other marked vertices on the path from $y$ to $v$.  See Figure~\ref{fig:marked-ancestor}(b).
\end{proof}

\begin{figure}
\centering
\begin{tabular}{c@{\hspace{2cm}}c}
\scalebox{1}{\includegraphics{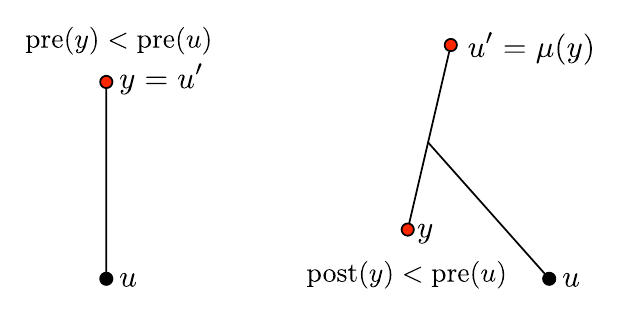}} &
\scalebox{1}{\includegraphics{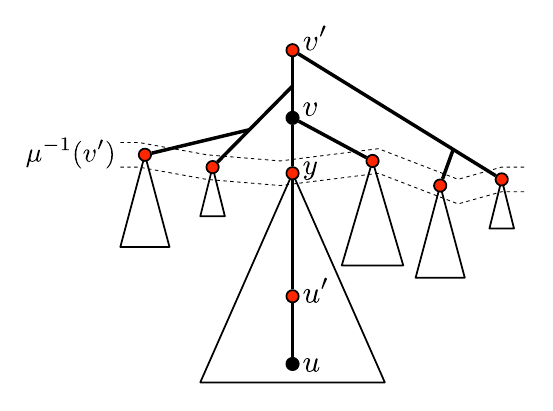}}\\
(a) & (b)
\end{tabular}
\caption{\label{fig:marked-ancestor}Marked vertices are shaded red.  (a) The nearest marked ancestor $u'$ of $u$
depends on whether the predecessor of $\pre(u)$ is $\pre(y)$ (left) or $\post(y)$ (right) for some marked $y$.
(b) $v'$ is the nearest marked ancestor of $v$.  The last marked ancestor $y$ on the $u$-to-$v$ path is in $\mu^{-1}(v')$.}
\end{figure}

We now have the following:

\begin{theorem}\label{Thm:EdgeFailures}
There is a data structure for planar undirected $n$-vertex graphs and any $d\geq 1$
that after $O(n)$ preprocessing time supports
\texttt{update}$(F)$ in $O(\sort(d,n))$ time, where $F\subseteq E$, $|F|\leq d$, and supports
\texttt{connected?}$(u,v)$ in $O(\pred(d,n))$ time.
\end{theorem}

\section{Vertex and Edge Failures}\label{sec:MultiVertexEdge}

We now turn our attention to the scenario where both edges and
vertices can fail.  The additional challenge arises from high degree
vertices that, when failed, can greatly reduce the connectivity of the
graph.  Nevertheless we will show how to maintain dynamic connectivity
in time proportional to the number of failed vertices, rather than
their degrees.  Formally, we develop a data structure that, given a
planar, undirected graph $G = (V,E)$, an integer $d \geq 1$, and a
dynamic subset of at most $d$ failed edges and vertices, supports the
following operations:
\begin{enumerate}
\item \texttt{update}$(F)$: set $F$ to be the failed set of vertices
  and edges; 
\item \texttt{connected?}$(u,v)$: are vertices $u$ and $v$
  connected in $G\setminus F$?
\end{enumerate}
As before we may assume that $G$ is plane and connected.

\subsection{Vertex failures for triconnected planar graphs}\label{sec:vft}
In this subsection we shall assume that $G$ is triconnected. This
allows us to apply Barnette's theorem, which states that every triconnected planar graph
has a spanning tree of degree at most three~\cite{Barnette66}. 
Furthermore, such a degree-three tree, $T$,
can be found in linear time~\cite{Strothmann97}.  
In Section~\ref{sect:reduction-to-tri} we show that this assumption is basically without loss of generality:
we can reduce the problem on general graphs to triconnected graphs with an additive $\pred(d,n)$ slowdown in the query and update algorithms.

Assume for simplicity that only vertices fail.  Section~\ref{subsec:MultiVertexEdge} describes
modifications needed when there is a mix of vertex and edge failures.
Let $C$ be the clockwise cycle that bounds the infinite and only face
of $T$.  $C$ is an embedding-respecting Euler tour of $T$ that visits
each edge twice and each vertex at most three times.  For a non-empty
subset $F$ of $V$, partition $C$ into maximal subpaths whose internal
nodes are not in $F$.  Denote this set of paths by ${\mathcal
P}_F$. Note that $|{\mathcal P}_F| \leq 3|F|$ and therefore that a connected component of $T\backslash F$
is made up of possibly many paths in ${\mathcal P}_F$.
Assign the connected components of $T \setminus F$ distinct colors 
and label each path in ${\mathcal P}_F$ with the color of its component.

Let $e_1$ and $e_2$ be the edges before and after a particular copy of
vertex $v$ in the order given by $C$.  Let $f_i$ be the face of $G$ to
the left of $e_i$. Root the dual spanning tree $T^*$ at the infinite
face of $G$ and let $\ell = \lca_{T^*}(f_1,f_2)$.  If $\ell$ is not
the root of $T^*$, let $e_v$ be the parent edge of $\ell$ in $T^*$.
Let $L$ be the set of edges $e_v$ (if well defined) with neither
endpoint in $F$ for all failed vertices $v\in F$ (according to their
multiplicity in $C$).  Note that $|L| \leq 3d$.  By duality, we consider
$L$ as a subset of primal edges.  Considered as an edge of the primal
graph, $e_v$ forms a cycle with $T$ that $v$ is on; that is, $e_v$ witnesses
an alternate connection should $v$ fail.

We define an auxiliary graph $H_F$ that will succinctly represent the
connectivity of ${\mathcal P}_F$. The nodes of $H_F$ are the paths in
${\mathcal P}_F$.  Two path-nodes $P_1, P_2 \in {\mathcal P}_F$ are
adjacent in $H_F$ iff
\begin{itemize}
\item $P_1$ and $P_2$ have the same color and are consecutive paths in
  $C$ among paths of the same color, or
\item there is an edge in $L$ between the interior of $P_1$
      and the interior of $P_2$. 
\end{itemize}
There are at most $|\mathcal P_F|$ edges of the first type and $|L|$
edges of the second type.
Since $|\mathcal P_F| \leq 3d$ and $|L| \leq 3d$, $|H_F| = O(d)$.
\begin{Lem}\label{Lem:ConnectedComps}
  Paths in $\mathcal P_F$ are connected in $G\setminus F$ iff they are
  connected in $H_F$.
\end{Lem}
\begin{proof}
  Consider distinct paths $P_1$ and $P_2$ in ${\mathcal P}_F$.  It is
  clear from the definition of $L$ and $H_F$ that if there is an edge
  $(P_1,P_2)$ in $H_F$ then the interior of $P_1$ and the interior of
  $P_2$ are connected in $G\setminus
  F$. This implies the ``if'' part.

  Since path-nodes of a given color class are connected by a cycle in
  $H_F$ given by the paths order along $C$.  If two paths are of the
  same color, their path-nodes will be connected in $H_F$.  So, for
  the ``only if'' part, it suffices to show that if $P_1$ and $P_2$ are
  different colors but (by transitivity of connectivity) there is a single
  edge $e\in G\setminus F$ between the interior of $P_1$ and the interior
  of $P_2$, then they are connected in $H_F$.

  So let $e =
  (u_1,u_2)$ be such an edge where $u_i\in P_i$.  Let
  $f_1$ and $f_2$ be the faces incident to $e$ in $G$ such that $f_1$
  is the child of $f_2$ in $T^*$. Let $f$ be the bounded face of
  $T\cup\{e\}$ (which contains $e$ and encloses $f_1$).  Let $v_i$ be
  the failed vertex that is an endpoint of $P_i$ on the bounded face
  of $T\cup\{e\}$.

  We continue by induction on the number of faces $k\geq 1$ of $G$
  contained in $f$. In the base case $k = 1$ and $f = f_1$.  Here
  $e_{v_1} = e_{v_2} = e$, so $e \in L$ and $(P_1,P_2)$ is an edge of
  $H_F$.  Now assume $k > 1$ and that the inductive hypothesis holds for
  smaller values. 

  If $f_1$ has a failed vertex $v$ on its boundary, then we argue $e_v
  = e$ (and so $(P_1,P_2)$ is an edge of $H_F$).  Consider the copy of
  $v$ in the traversal of $C$ that has $f_1$ to the left.  Let $g_1$
  and $g_2$ be the faces to the left of the edges before and after
  this copy of $v$ in $C$.  Since $e \in T^*$, $e$ is an ancestral edge
  of $g_1$ and $g_2$.  Since $g_1,g_2,f_1$ all contain $v$ as a
  boundary vertex, $\lca_{T^*}(g_1,g_2)=f_1$.

Suppose that $f_1$ has no failed vertices on the boundary.  For each edge $e'\neq e$ on the boundary of $f_1$, 
if $e'\in T$ then $e'$ belongs to a path of $\mathcal P_F$. Otherwise, the bounded face of
 $T\cup\{e'\}$ is contained in $f$ and does not contain $f_1$ so it
 contains fewer than $k$ faces of $G$. Then the paths of $\mathcal
 P_F$ containing the endpoints of $e'$ in their interior are
connected in $H_F$ if they
 have distinct colors (by the inductive argument) or if they have the
 same color (by the start of this proof).  Since this holds for every
 edge $e'$ of $f_1\setminus\{e\}$, $P_1$ and $P_2$ are connected in
 $H_F$, as desired.
\end{proof}

\subsection{The data structure}
In a linear-time preprocessing step, we compute $T$, $T^*$, and $C$
and initialize an $\lca$ data structure.

\paragraph{\texttt{update}$(F)$:} We build $H_F$ for the input set of
failed vertices $F$.  Let $f$ be the sum of the degrees in $T$ of the
vertices in $F$; note, $f \leq 3d$.  Let $F'$ be the multiset of
failed vertices according to their multiplicity along $C$, i.e., their 
degree in $T$.  Again, $|F'| = f$.  Sort the vertices of $F'$ according to
their order along $C$.  This provides an implicit representation of
${\mathcal P}_F$.  Label these paths according to their ordinal in
$F'$; that is, path $P_i$ is the path starting with the $i^{th}$
vertex in sorted $F'$.  This takes time $\sort(f,n)$.

Greedily assign colors to the paths, considering the paths in order.
In each iteration we assign colors to all the paths in a given color
class.  Upon considering path $P_i$, if $P_i$ has not yet been
colored, assign it a new color.  Check the second last endpoint of
$P_i$ and find the edge $e$ that precedes that vertex in $C$ {\em
  after} $P_i$.  Let $P_j$ be the path that starts with edge $e$
(which can be found by $e$'s starting point in $C$).  Color $P_j$ with
same color as $P_i$.  Repeat until returning to $P_i$.  This takes
time $O(f)$.

Build the set $L$ of edges.  For each edge in $L$, identify the paths in $\mathcal{P}_F$
that contain its endpoints.  We do this in bulk.  Sort the set of
endpoints of $L$ along with $F'$, that is, $V(L) \cup F'$, according to their
order along $C$.  Traverse this order, assigning each endpoint in
$V(L)$ the path corresponding to the last failed vertex visited.  This
takes time $\sort(f,n)$.

From $L$ (with endpoints labelled with the appropriate path in
${\mathcal P}_F$) and the colors of ${\mathcal P}_F$, build $H_F$.
Compute the connected components of $H_F$ and label the path-nodes of
$H_F$ with the name of the connected component it belongs to.  This
takes time $O(f)$.

The total time for {\tt update} is bounded by $\sort(f,n) =
O(\sort(d,n))$.

\paragraph{\texttt{connected?}$(u,v)$:} We may assume $u,v\notin F$.  In $O(\pred(d,n))$ time 
identify any path $P_u$ containing $u$ in the interior and any path $P_v$ containing $v$ in the interior.  
By Lemma~\ref{Lem:ConnectedComps}, $u$ and $v$ are
connected in $G\setminus F$ iff $P_u$
and $P_v$ are in the same connected component of $H_F$.  The latter
condition is checked in constant time given the labels of the connected components.

\subsection{Dealing with vertex and edge failures simultaneously}\label{subsec:MultiVertexEdge}
We now extend our results above for triconnected planar graphs to handle both vertex and edge failures.
As before, we precompute $T$, $T^*$, $C$, and a structure for answering lca queries. 
An \texttt{update}-operation now gets as input
a set $F\subseteq V\cup E$ of size
at most $d$. For the failures in $F\cap V$, we partition the tree into set $\mathcal P_F$ of paths as before. Edges of $F$ that
belong to $T$ can be treated in a way similar to vertex failures (this can be seen easily by, say, introducing a degree two-vertex
to the interior of each such edge and regarding it as a failed vertex).

What remains is to deal with failed edges that go between paths from $\mathcal P_F$. The only modification we need
to make to \texttt{update} is that when we find the $\lca$-vertex $l$ in $T^*$ of two faces of $G$ incident to a failed vertex or edge
of $T$, we need to check if the edge from $l$ to its parent in $T^*$ has failed. If so, we need to walk up the path from $l$ to the
root $r$ of $T^*$ until we find a vertex $l'$ whose parent edge has not failed or $l' = r$. Our algorithm then has the
same behaviour as the \texttt{update} operation of the previous subsection would have when applied to $G\setminus (F\cap E)$ with
input $F\cap V$.

The only issue is the additional running time needed to traverse paths to $r$ in $T^*$. We deal with this as follows. When
a path has been traversed, we store pointers from all visited vertices to the last vertex on the path (i.e., the vertex closest to
the root of $T^*$). If in another traversal
we encounter a vertex with a pointer, we use this pointer to jump directly to that last vertex. The total number
of edges traversed in $T^*$ is then bounded by the number of failed edges which is at most $d$. Hence, \texttt{update} still runs in
$\sort(d,n)$ time. The \texttt{connected?}-operation can be implemented as before.

By Theorem~\ref{thm:reduction-to-tri} the algorithm above extends to general planar graphs with no asymptotic slowdown.

\begin{theorem}\label{Thm:VertexEdgeFailures}
There is a data structure for any planar undirected $n$-vertex graph that, after $O(n)$ preprocessing time, supports
\texttt{update}$(F)$ in $O(\sort(d,n))$ time, where $F\subseteq V\cup E$, $|F|\leq d$, and supports
\texttt{connected?}$(u,v)$ in $O(\pred(d,n))$ time.
\end{theorem}

\section{Decremental Subgraph Connectivity}\label{sec:SingleVertexEdge}

In this section, we consider the model in which updates and queries
are intermixed. We only allow failures and not recoveries but no longer assume
an upper bound $d$ on $|F|$, that is, $F$ is a growing set.

As in Section~\ref{sec:MultiVertexEdge}, we first assume that $G$ is
triconnected and that only vertices can fail. Define $T$, $T^*$, $C$,
$L$, and $F$ as before and build an $\lca$ data structure for $T^*$.
Initially $L$ and $F$ are empty.  Redefine $H_F$ to be $(V,T\setminus F \cup L)$, that is,
we no longer contract paths in $\mathcal{P}_F$.
For each vertex $v$ we maintain a list $L(v)$ of the $\lca$-edges $L$
incident to $v$.  Represent $L$ as a subset of marked edges of $G$.

When a vertex $v$ fails, unmark $L(v)$ in $L$ since $L$ should only
contain edges not incident to failed vertices. For each of the
at-most-three occurrences of $v$ along $C$, find the corresponding
$\lca$-edge $e_v$ (if it exists); if $e_v$ has no endpoint in $F$,
mark $e_v$ as a member of $L$ and add $e_v$ to the lists of its
endpoints.

It follows easily from the definition of $L$ and of $H_F$ that both
$L$ and $H_F$ are updated correctly after each
failure. Lemma~\ref{Lem:ConnectedComps} (that two vertices are
connected in $G\setminus F$ iff they are connected in $H_F$) holds
with these definitions; the only difference is the interior of the
paths from ${\mathcal P}_F$ are explicitly represented with the edges
from $T\setminus F$.  Note that showing the requisite connectivity for
Lemma~\ref{Lem:ConnectedComps} does not depend on $\lca$-edges that
have a failed endpoint.  Removing (unmarking) the edges $L(v)$ as
members of $L$ when $v$ fails is equivalent to not including them in
$L$ as used in the proof of Lemma~\ref{Lem:ConnectedComps}.

We maintain $H_F$ with an oracle which allows us to delete an edge and
an isolated vertex in $O(\log n)$ 
time~\cite{EppsteinITTWY92} since a vertex failure results in at
most three edges of $T$ being deleted from $H_F$. Furthermore, at most three
edges are added to $L$ and each of these edges is added only to the
two lists associated with its endpoints. Hence if $F$ is the final set
of deleted vertices, then a total of $O(|F|)$ edges are added or
removed from $H_F$ and there are at most $O(|F|)$ updates to $L$ and
to the lists $L(v)$. It follows that each deletion takes $O(\log n)$
amortized time. A connectivity query can be answered in $O(\log n)$
worst-case time with the oracle associated with $H_F$.

So far we considered only vertex deletions. Edge deletions can be
handled in the same way as in Section~\ref{subsec:MultiVertexEdge}. 
Using the reduction from general planar graphs to triconnected planar graphs in Theorem~\ref{thm:reduction-to-tri-fully-dynamic}
we have the following. 

\begin{theorem}
  There is a data structure for decremental subgraph connectivity in triconnected planar graphs
  that, after $O(n)$ preprocessing time, allows an intermixed sequence of vertex/edge failures
  and connectivity queries in $O(\log n)$ amortized time per failure (i.e., $f$ failures in $O(f\log n)$ time total) and $O(\log n)$ worst-case time per query.
  In general planar graphs there is a data structure requiring $O(\log^2 n)$ amortized time per failure and $O(\log^2 n)$ worst-case time per query.
\end{theorem}

\section{Reductions to Triconnected Graphs}\label{sect:reduction-to-tri}

In this section we give the reductions from general (planar) graphs to triconnected (planar) graphs that were used in 
Sections~\ref{sec:MultiVertexEdge} and \ref{sec:SingleVertexEdge}.  Although the algorithm from Section~\ref{sec:SingleVertexEdge} handles only failures, the reduction allows both failures and recoveries.

\begin{theorem}\label{thm:reduction-to-tri} {\bf (The $d$-Failure Model)}
Suppose there is a connectivity oracle $\mathcal{A}_3$ for any triconnected (planar) graph $G$ with the following parameters.
The preprocessing time is $P_3(n,m)$, 
the time to update the structure after $d$ edge or vertex failures
is $U_3(d,n)$, and the time for a connectivity query is $Q_3(d,n)$.  Then there is a connectivity oracle $\mathcal{A}$
for any (planar) graph $G$ with parameters $P = O(P_3)$, $U = O(U_3 + \sort(d,n) + d\cdot (Q_3+\pred(d,n)))$, and $Q = O(Q_3 + \pred(d,n))$.
\end{theorem}

\begin{theorem}\label{thm:reduction-to-tri-fully-dynamic} {\bf (The Fully Dynamic Subgraph Model)}
Suppose there is a connectivity oracle $\mathcal{A}_3$ for any
triconnected (planar) graph $G$ with the following parameters.
The preprocessing time is $P_3(n,m)$, the time to process a vertex
failure or recovery is $U_3^{v}(n)$,
the time to process an edge failure or recovery $U_3^e(n)$, and the
time for a connectivity query is $Q_3(n)$.
Then there is a connectivity oracle $\mathcal{A}$ for any (planar)
graph $G$ with parameters
$P = O(P_3)$, $U^e = O((U_3^e + Q_3 + \log n/\log\log)\log n)$, $U^v = O(U_3^v
+ (U_3^e + Q_3 + \log n/\log\log n)\log n)$,
and $Q = O((Q_3 + \log n/\log\log n)\log n)$.
\end{theorem}

The proof of Theorem~\ref{thm:reduction-to-tri} requires two steps: a reduction from connected to biconnected graphs
and a reduction from biconnected to triconnected graphs.  These are given in Sections~\ref{sect:gen-to-bi} and~\ref{sect:bi-to-tri}.
Theorem~\ref{thm:reduction-to-tri-fully-dynamic} is proved in Section~\ref{sect:reduction-to-tri-fully-dyn}.
Throughout this section we assume for simplicity that $F\subseteq V(G)$ consists solely of vertices.

\subsection{Reduction to Biconnected Graphs in the $d$-Failure Model}\label{sect:gen-to-bi}

Let $P_2,U_2,$ and $Q_2$ be the preprocessing, update, and query times for an oracle for biconnected graphs.
Let $A\subseteq V(G)$ be the set of articulation points (vertices
whose removal disconnects $G$) and $\mathcal{C}$ be the set of
biconnected components in $G$.  The {\em component tree}
$\mathcal{T}_2$ on the vertex set $A\cup\mathcal{C}$ consists of edges
connecting articulation points to their incident biconnected
components.\footnote{The subscript 2 refers to biconnectivity.}  Root
$\mathcal{T}_2$ at an arbitrary $A$-node and let $\rep(C)\in A$ be the
parent of $C\in\mathcal{C}$ in $\mathcal{T}_2$.  The function $\phi$
maps graph vertices into their corresponding $\mathcal{T}_2$ nodes,
that is, if $u\in A$, $\phi(u) = u$; otherwise, $\phi(u) = C$ where
$C$ is the biconnected component containing $u$.

For any $u\in V(G)\backslash F$ let $\rep^{**}(u)$ be the {\em most} ancestral $A$-node reachable from $u$ in $G\backslash F$,
or $u$ if there is no such $A$-node, that is, if $u\in C \in \mathcal{C}$ is disconnected from $\rep(C)$.
The goal of the preprocessing algorithm is to build a structure that lets us evaluate $\rep^{**}$.

\paragraph{Processing Vertex Failures}
First, mark {\em red} every articulation point in $F\cap A$ and {\em blue} every component $C$ for which $|V(C)\backslash\{\rep(C)\} \cap F| \neq \emptyset$.  This takes $O(d)$ time.  
For a $\mathcal{T}_2$-node $u$, let $\rep^*(u)$ be the most distant $A$-node ancestor of $u$ reachable by a path of uncolored vertices.  By Lemma~\ref{Lem:marked-ancestor} we can evaluate $\rep^*(u)$ in $\pred(d,n)$ time with $\sort(d,n)$ preprocessing.  (Note that $\rep^*(u)$ does not exist if $u\in A$ is red or if $u\in \mathcal{C}$ and $\rep(u)$ is red.)  

In $O(U_2(d,n))$ time preprocess each blue component $C$ to support connectivity queries in $V(C)\backslash F$.
If $\rep(C)$ is uncolored, let $\hat C$ be the parent of $\rep^*(C)$ in $\mathcal{T}_2$.
(It may be that $\rep(C) = \rep^*(C)$.)
Determine if $\rep^*(C)$ and $\rep(\hat C)$ are still connected in $V(\hat C)\backslash F$ via one connectivity query and, if not, 
mark {\em red} the edge $(\rep^*(C),\hat C)$.  A red edge indicates that $\rep^*(C)$ is disconnected from $\rep(\hat C)$ but possibly still connected to other vertices in $\hat C$.  It takes $O(d\cdot Q_2(d,n))$ time to find all such red edges.
For $u\in A\cup\mathcal{C}$ let $\rep^{**}(u)$ be the most distant $A$-node ancestor of $u$ reachable by a path of uncolored vertices and edges.  
Using a preorder traversal of the blue nodes, compute $\rep^{**}(\rep(C))$ for all blue $C$ in $O(d)$ time.

\paragraph{Connectivity Queries}

To evaluate $\rep^{**}(x)$ and $\rep^{**}(x')$ we require two intra-component connectivity queries and two $\rep^*$ queries.
To determine if $x$ and $x'$ are connected we may require an additional intra-component connectivity query,
for a total time of $O(Q_2(d,n) + \pred(d,n))$.

Let $x_1 = \phi(x)$ and $x_1'=\phi(x')$ be the corresponding nodes in $\mathcal{T}_2$.  
Let $x_2 = \rep^*(x_1)$ and, if $\rep^*(x_1)$ exists, let $C_2$ be the parent of $x_2$.  
If $x_2$ does not exist let $x_3 = x$ and let $C_3$ be $x$'s component.
If $x_2$ exists and is connected to $\rep(C_2)$ let $x_3 = \rep^{**}(\rep(C_2))$, 
otherwise let $x_3 = x_2$; in both cases let $C_3$ be the parent of $x_3$, if $x_3$ is not already the root.
It follows that $x_3$ is a graph node, not a component, that $x$ is connected to $x_3$ in $G\backslash F$, 
and that $x_3$ is disconnected from $\rep(C_3)$.
The nodes $C_1',x_2',C_2',x_3',$ and $C_3'$ are defined similarly.  
It follows that we must return {\em connected} if $x_3 = x_3'$,
{\em disconnected} if $C_3 \neq C_3'$,
and {\em connected} if $x_3$ and $x_3'$ are connected in $C_3 = C_3'$.
In total there are up to three intra-component connectivity queries, on the pairs $(x_2,\rep(C_2)), (x_2',\rep(C_2')),$ and $(x_3,x_3')$.
There are two necessary queries to $\rep^*(x_1)$ and $\rep^*(x_1')$.  The nodes $\rep^{**}(\rep(C_2))$ and $\rep^{**}(\rep(C_2'))$ were computed when processing $F$.

\begin{figure}
\centering
\begin{tabular}{c@{\hspace{.8cm}}c@{\hspace{.8cm}}c}
\scalebox{1}{\includegraphics{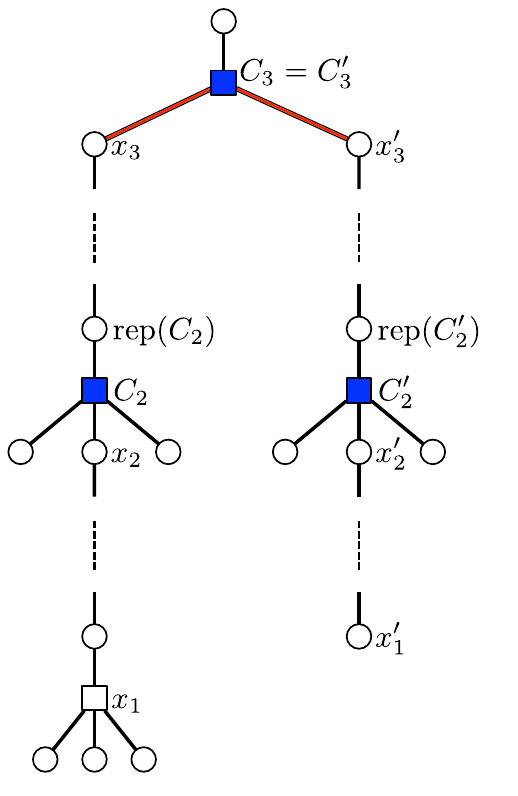}}
&
\scalebox{1}{\includegraphics{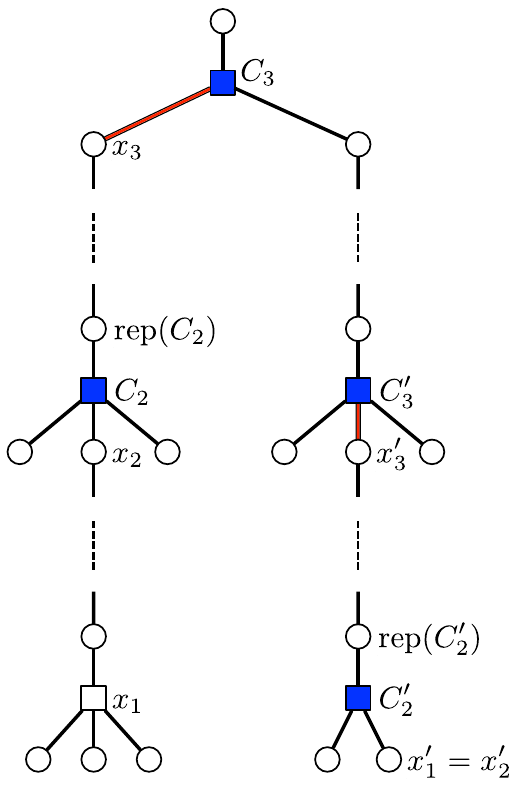}}
&
\scalebox{1}{\includegraphics{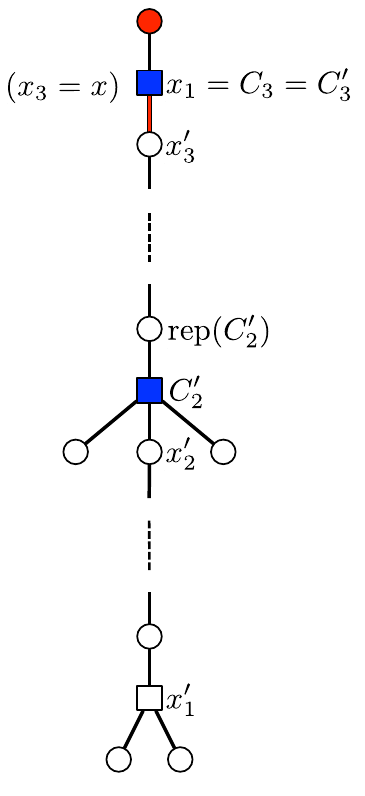}}\\
(a) & (b) & (c)
\end{tabular}
\caption{\label{fig:gen-to-bi}Three examples of queries.
Square nodes represent biconnected components, round nodes are articulation points.  
A round node shaded red is a failed articulation point.  A square node shaded blue is a component with at least 1 failure.  An edge $(z,C)$ shaded red indicates that $z$ is disconnected from $\rep(C)$, but not necessarily from other nodes in $C$.}
\end{figure}

\subsection{Reduction to Triconnected Graphs in the $d$-Failure Model}\label{sect:bi-to-tri}

An $SPQR$ tree $\mathcal{T}_3$ represents how a biconnected graph $G$ is assembled form its triconnected components through {\em merging} operations.  All separating pairs of $G$ are given implicitly by $\mathcal{T}_3$.  Let $H_1$ and $H_2$ be two edge-labeled 
multigraphs, where $V(H_1) \cap V(H_2) = \{x,y\}$ and where $e=(x,y)$ appears in both $E(H_1)$ and $E(H_2)$ with the same label.  
Then {\em merging} $H_1$ and $H_2$ results in a new graph $H_3 = (V(H_1)\cup V(H_2), E(H_1) \cup E(H_2) \backslash\{e\})$.  
That is, $e$ does not appear in the merged graph but $x$ and $y$ may still be adjacent in $H_3$ if there were multiple edges between $x$ and $y$ in either $H_1$ or $H_2$.

Each node $u\in \mathcal{T}_3$ is associated with a triple $\ang{M_u,G_u,\pi(u)}$,
where $M_u$ is the {\em model} graph, one edge $\pi(u)\in E(M_u)$ is distinguished as the {\em polar edge} whose endpoints are {\em poles}.  
Let $p(u)$ be the parent of $u$ in $\mathcal{T}_3$ and $\mathcal{T}_3(u)$ the subtree rooted at $u$.  
The graphs $M_u$ and $M_{p(u)}$ have exactly two vertices in common, namely the endpoints of $\pi(u)$, and each contains identically labeled copies of $\pi(u)$.  Moreover, for any two $\mathcal{T}_3$-nodes on opposite sides of the edge $(u,p(u))$, their model trees do not intersect, except possibly at the poles of $\pi(u)$.
The graph $G_u$ is formed by merging all model trees in $\mathcal{T}_3(u)$.  That is, $u$ 
has children $u_1,\ldots,u_k$ where $V(G_{u_i}) \supseteq V(M_{u_i})$ and $G_{u_i}$ contains the polar edge $\pi(u_i)$.
Obtain $G_u$ by merging $\{G_{u_i}\}$ with $M_u$ along the polar edges $\{\pi(u_i)\}$.  (Observe that the polar edge $\pi(u) \in E(M_u)$ remains in $E(G_u)$.)
Let $G_u^-$ be $G_u$ without its polar edge $\pi(u)$.

The tree $\mathcal{T}_3$ is conceptually constructed in a top-down fashion as follows.  (Linear time triconnectivity algorithms~\cite{HopcroftT73,MillerR92,GutwengerM00} can be used to construct $\mathcal{T}_3$.)  The root $r$ is associated with $G_r = G$ and an arbitrary polar edge $\pi(r) \in E(G)$.  
In general, we are given a $\mathcal{T}_3$-node $u$ and pair $\ang{G_u,\pi(u) \in E(G_u)}$.  
We must select a model multigraph $M_u$ and then recursively construct the subtree $\mathcal{T}_3(u)$.
The reader may want to refer to an illustration of a graph and its $SPQR$-tree in Figure~\ref{fig:spqr}.

\begin{figure}
\centering
\begin{tabular}{c}
\scalebox{.8}{\includegraphics{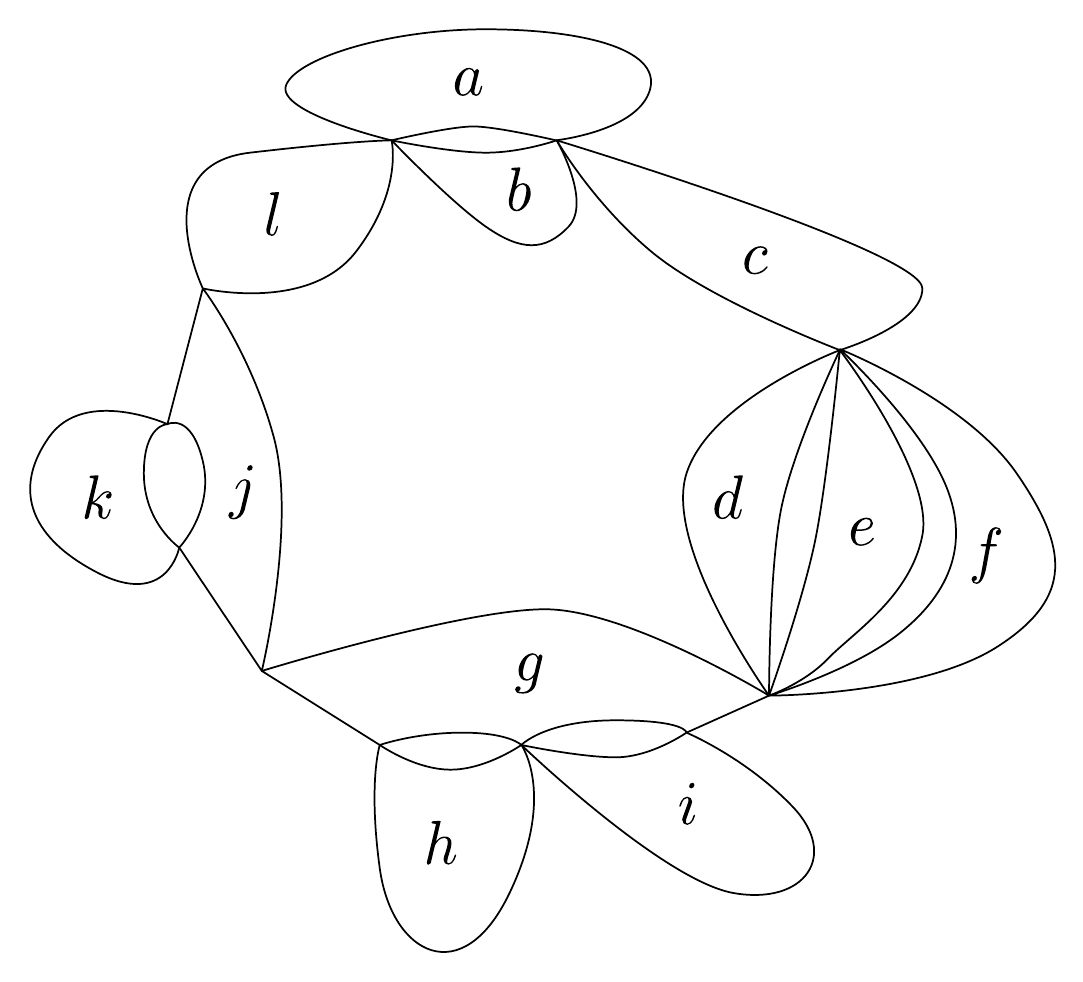}}
\\
\\
\scalebox{1.2}{\includegraphics{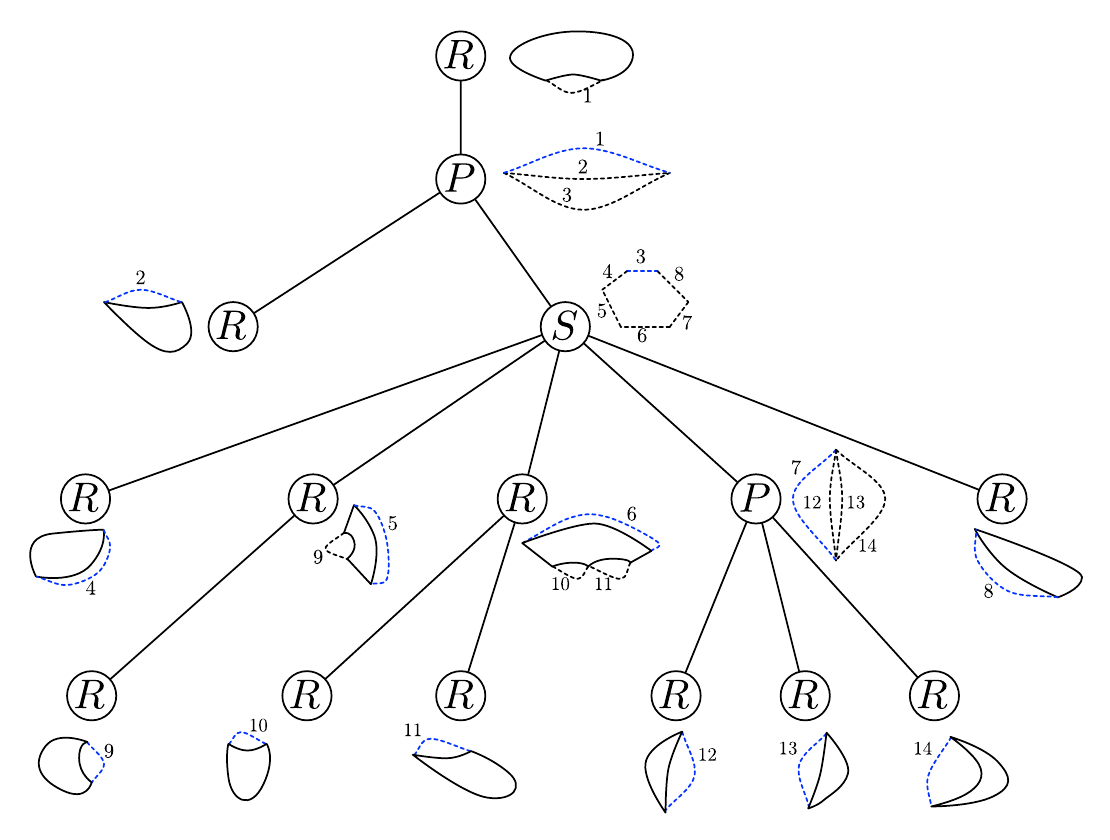}}
\end{tabular}
\caption{\label{fig:spqr}Above: a schematic of a biconnected graph; vertices internal to subgraphs $a,\ldots,l$ are not shown.
Below: The corresponding $SPQR$-tree $\mathcal{T}_3$, with mergeable edges labeled.  Adjacent to each $u\in \mathcal{T}_3$ is its model graph $M_u$.}
\end{figure}

\paragraph{Parallel Case ($P$-node).}
The endpoints of $\pi(u) = (\pi_0,\pi_1)$ split $E(G_u)$ into $k+1$ components $\pi(u), G_{u_1}^-,\ldots,G_{u_k}^-$ where $k\ge 2$.
Let $G_{u_i}$ be $G_{u_i}^-$ augmented with a freshly labeled polar edge $\pi(u_i) = (\pi_0,\pi_1)$ and let $M_u$ consist of $k+1$ parallel edges $\{\pi(u),\pi(u_1),\ldots,\pi(u_k)\}$.  Clearly merging $\{G_{u_i}\}$ with $M_u$ yields $G_u$.  For each nontrivial $G_{u_i}^-$ (nontrivial means having more than one edge), create a new child $u_i$ of $u$ and recursively compute its subtree in $\mathcal{T}_3$.\footnote{In previous work on $SPQR$ trees~\cite{DiBattistaT96} trivial graphs are actually assigned to children of $u$.  They are called $Q$-nodes.}

\paragraph{Series Case ($S$-node).}
Let $\pi(u) = (\pi_0,\pi_1)$ and let $\pi_2,\ldots,\pi_k$ be the articulation points of $G_u^-$.  If there are any such points ($k\ge 2$)
let $\pi(u),G_{u_1}^-,\ldots,G_{u_k}^-$ be the partition of $G_u$ into edge-disjoint subgraphs such that $G_{u_i}^-$ and $G_{u_{i+1}}^-$ meet at $\pi_{i+1}$.  Let $M_u$ be the cycle $(\pi_0,\ldots,\pi_k,\pi_0)$ and let $G_{u_i}$ be $G_{u_i}^-$ augmented with a polar edge $(\pi_i,\pi_{i+1})$ (where $\pi_{k+1} = \pi_0$), whose label matches the corresponding edge in $M_{u_i}$.  For each non-trivial $G_{u_i}^-$ form a new child $u_i$ of $u$ and recursively compute its subtree in $\mathcal{T}_3$.

\paragraph{Rigid Case ($R$-node).}
The previous two cases do not apply.  A separating pair $\alpha = \{\pi_0',\pi_1'\}$ of $G_u$ splits it into some number of components $H_1,\ldots,H_l$ with $\pi(u) \in E(H_1)$.  Let $H(\alpha) = H_1$ and $H'(\alpha)$ be the union of $H_2,\ldots,H_l$.  Call $\alpha$ {\em maximal} if there is no separating pair $\beta$ for which $H(\beta)$ is strictly contained in $H(\alpha)$.
Let $\{\pi(u_i)\}_{1\le i \le k}$ be the set of $k$ maximal separating pairs with respect to $\pi(u)$.
We give $u$ children $u_1,\ldots,u_k$ where $u_i$ is given the graph $G_{u_i} = H'(\pi(u_i))$ and polar edge $\pi(u_i)$.
The model graph $M_u$ is induced by $E(M_u) = E(H(\pi(u_i))) \cap \cdots \cap E(H(\pi(u_k))) \cup \{\pi(u_1),\ldots,\pi(u_k)\}$, that is,
it is obtained from $G_u$ be replacing $H'(\pi(u_i))$ with the $u_i$'s polar edge $\pi(u_i)$, for all $i$.
The maximality of the separating pairs ensures that $M_{u}$ is triconnected, also called {\em rigid}.

\begin{Obs}\label{obs:PSnodes}
In $\mathcal{T}_3$ there are no adjacent $P$-nodes nor adjacent $S$-nodes.
\end{Obs}

The edges in a model graph $M_u$ that do not appear in $G$ are {\em artificial} and represent paths in $G$.  
Our goal is to reduce a connectivity query in $G$ to a constant number of connectivity queries in model graphs, so it is important that we identify the set of invalidated artificial edges.  

\paragraph{Preprocessing Vertex Failures}

Let $\phi(z)$, where $z\in V(G)$, be the {\em most} ancestral $\mathcal{T}_3$-node $u$ for which $z\in V(M_u)$ and let $\phi^{-1}(u) = \{z \::\: \phi(z) = u\}$.
(Note that $z$ can appear in an unbounded number of model trees as a pole.)  We begin by coloring {\em blue} those $\mathcal{T}_3$-nodes $u$ for which the endpoints of $\pi(u)$ may be disconnected in $G_u^- \backslash F$.  Specifically, repeatedly color blue any $u$ satisfying one of the following criteria:

\begin{enumerate}
\item $u = \phi(z)$ for some $z\in F$.
\item $u$ is an $S$-node with at least one blue child.  ($M_u\backslash{\pi(u)}$ is a path, so the removal of one vertex or edge  disconnects $u$'s poles.)
\item $u$ is either an $R$- or $P$-node with at least two blue children.  (There are at least two internally vertex-disjoint paths
between the endpoints of $\pi(u)$ in $M_u\backslash{\pi(u)}$, so the removal of just one vertex or edge cannot disconnect $u$'s poles.)
\end{enumerate}

\begin{Lem}
There are fewer than $4d$ blue nodes.  They can be found in $O(d)$ time.
\end{Lem}

\begin{proof}
There are clearly at most $d$ blue nodes due to criterion (1) and at most $d-1$ due to criterion (3).
Since, by Observation~\ref{obs:PSnodes}, there are no adjacent $S$-nodes, the number due to criterion (2)
is no more than the number due to criteria (1) and (3).  They can be found in $O(d)$ time as follows.  For each $z\in F$
color $\phi(z)$ blue and continue to color any immediate ancestors blue if they satisfy criteria (2) or (3).  The first ancestor not colored blue is noted as having one blue child.  The time per $F$-node is $O(1)$ plus the number of formerly uncolored ancestors colored blue.
\end{proof}

Next we color {\em purple} all blue nodes $u$ for which the poles of $\pi(u) = (\pi_0,\pi_1)$ are disconnected in $G_u^- \backslash F$ (that is, ignoring parts of the graph outside $u$'s subtree), as follows.  
Traverse the blue nodes in $\mathcal{T}_3$ in postorder.  If either of $\pi_0$ or $\pi_1$ is in $F$ then they are trivially disconnected, so color $u$ purple.  If $u$ is a $P$-node then $\pi_0$ and $\pi_1$ are still connected if they are connected in {\em any} child of $u$, so color $u$ purple if all its children are purple.  If $u$ is an $S$-node then $\pi_0$ and $\pi_1$ are disconnected if the poles of any child $u_i$ of $u$ are disconnected, so color $u$ purple if any child is purple.  Finally, if $u$ is an $R$-node, construct a connectivity oracle $\mathcal{A}_u^-$ for $M_u$ avoiding vertex failures $F\cap V(M_u)$ and edge failures $\{\pi(u)\} \cup \{\pi(u') \;:\; \mbox{$u'$ is a purple child of $u$}\}$.
(The edge failures reflect the fact that the polar edges of purple children are definitely invalidated by the failures and 
it is not yet known if $\pi(u)$ is invalidated.)  Perform one query on $\mathcal{A}_u^-$ to determine if $\pi_0,\pi_1$ are disconnected in $G_u^- \backslash F$ and, if so, color $u$ purple.

We also construct connectivity oracles for each colored $S$- and $P$-node $u$ in $G^-_u\backslash F$.
This is trivial if $u$ is a $P$-node. ($M_u$ contains two vertices, which are either connected or not.)  
If $u$ is an $S$-node then $M_u$ is a cycle and a connectivity query can be answered with predecessor search, in $\pred(d,n)$ time.
The time to construct oracles for all colored nodes is $O(U(d,n) + \sort(d,n) + d\cdot Q(d,n))$.

\begin{Lem}
In the construction of connectivity oracles $\{\mathcal{A}_u^-\}$, there are $O(d)$ vertex and edge failures {\em in total.}
\end{Lem}

\begin{proof}
Each colored node $u$ can contribute one edge failure to $\mathcal{A}_{p(u)}^-$, namely the polar edge $\pi(u)$.
Note that $F\cap V(M_u)$ can include at most two additional vertices than $\{z \in F \;:\; \phi(z) = u\}$, namely the poles
of $\pi(u)$.  Thus, the total number of failures is linear in $d$ and the number of colored nodes, which are fewer than $4d$.
\end{proof}

Let $\rep^{**}(z)$ be the farthest ancestral pole reachable from $z\in V(G)$, or $z$ if there is no such pole.
Specifically, let $u\in \mathcal{T}_3$ be the farthest ancestor of $\phi(z)$ such that $z$ is connected to either pole
in $\pi(u)$ in $G_u^-\backslash F$.  If there is no such $u$ ($z$ is disconnected from $\pi(\phi(z))$) then $\rep^{**}(z) = z$
and if $z$ is connected to both poles $\pi(u)$ then $\rep^{**}(z)$ can be either one.
Our goal is to preprocess $F$ so that we can quickly evaluate $\rep^{**}$.

Let $\rep(z)$ be either pole in $\pi(\phi(z))$ connected to $z$ in $G_{\phi(z)}^-$, or undefined if there is no such pole.
Let $\rep^*(z)$ be the farthest ancestral pole reachable from $z$ via uncolored $\mathcal{T}_3$ nodes, defined precisely as follows.  If $\phi(z)=u$ and $u$ is colored then $\rep^*(z)$ is undefined.  If $u$ is uncolored and $z$ is disconnected from both of $u$'s poles (this can only happen if they are in the failure set) then $\rep^*(z)$ is undefined.
Otherwise let $\hat u$ be the farthest ancestor of $u$ such that the $\mathcal{T}_3$-path from $u$ to $\hat u$ is uncolored and at least one pole of $\hat u$ is not in $F$.  Then $\rep^*(z)$ is any non-$F$ pole of $\hat u$.  By Lemma~\ref{Lem:marked-ancestor}, $\mathcal{T}_3$
can be preprocessed in $\sort(d,n)$ time to support $\rep^*$ queries in $\pred(d,n)$ time.

We calculate $\rep^{**}$ on each pole $z$ of each colored $\mathcal{T}_3$-node $u$ by a preorder traversal of such nodes,
as follows.
\begin{enumerate}
\item If $z\in F$ then $\rep^{**}(z)$ is undefined.
\item If $u$ is the root of $\mathcal{T}_3$ or is a colored node without a colored ancestor then $\rep^{**}(z)$ is a pole of the root of $\mathcal{T}_3$.
\item If $\phi(z)$ is colored then let $\rep^{**}(z) = z$ if $\rep(z)$ is undefined (it is disconnected from both poles of $\phi(z)$)
and $\rep^{**}(\rep(z))$ otherwise.

\item If $\phi(z)$ is uncolored then let $\rep^{**}(z) = z$ if $\rep^*(z)$ is undefined (it is disconnected from both poles of $\phi(z)$)
or $\rep^*(z)$ if $\rep(\rep^*(z))$ is undefined ($\rep^*(z)$ is disconnected from both poles of $\phi(\rep^*(z))$)
or $\rep^{**}(\rep(\rep^*(z)))$ otherwise.
\end{enumerate}

Note that we calculate $\rep^{**}$ explicitly on only $O(d)$ vertices, which takes 
$O(d\cdot \max\{Q_3(d,n), \pred(d,n)\})$ time.

\paragraph{Connectivity Queries}

The query asks whether $x$ and $x'$ are connected in $G\backslash F$, where $x,x' \not\in F$.
Our goal is to reduce this to $O(1)$ connectivity queries to the oracles associated with colored $\mathcal{T}_3$-nodes.

Let $x_1 = \rep^*(x)$ if $\phi(x)$ is uncolored and $x$ if $\phi(x)$ is colored.
It follows that $\phi(x_1) = u_1$ is colored.
Using two connectivity queries on $u_1$'s oracle 
(between $x_1$ and the poles of $u_1$) determine $x_2 = \rep(x_1)$. 
If $x_2$ is defined let $x_3 = \rep^{**}(x_2)$;
otherwise let $x_3 = x_1$.
It follows directly from the definitions of $\rep,\rep^*,$ and $\rep^{**}$ that $x_3= \rep^{**}(x)$ is the most ancestral pole connected to $x$, or is $x$ if there is no such pole.
In any case let $u_3 = \phi(x_3)$ and define $x_1',x_2',x_3',u_3'$ analogously, with respect to $x'$.

\begin{Lem}\label{lem:}
Vertices $x$ and $x'$ are connected in $G\backslash F$ if and only 
if $u_3 = u_3'$ and $x_3$ and $x_3'$ are connected in $G\backslash F$.
It requires $O(Q_3(d,n) + \pred(d,n))$ time to answer a connectivity query.
\end{Lem}

\begin{proof}
Note that by definition of $\rep,\rep^*,$ and $\rep^{**}$, $x$ is connected to $x_1$,
$x_1$ is connected to $x_2$ (if it exists), and $x_2$ is connected to $x_3$, hence
$x$ is connected to $x_3$.
Thus, $x$ and $x'$ are connected if and only if $x_3$ and $x_3'$ are connected,
which is determined by $u_3$'s oracle if $u_3 = u_3'$.  (Note that $u_3$ is either colored and we have such an oracle, 
or $u_3$ is the uncolored root and $x_3 = x_3'$.)

Observe that if $u_3$ is not the root, $x_3$ is, by definition of $\rep^{**}$, disconnected from the poles $\pi(u_3)$, and therefore also disconnected from {\em every} vertex in $V(G) \backslash V(G_{u_3})$.  The same is true of $x_3'$
and $V(G) \backslash V(G_{u_3'})$.  If $u_3 \neq u_3'$ then either $x_3' \not\in V(G_{u_3})$ or $x_3 \not\in V(G_{u_3'})$,
that is, $x_3$ and $x_3'$ must be disconnected.

Concerning the time bound, 
we spend $O(\pred(d,n))$ time on queries to $\rep^*$ and $O(\max\{Q(d,n),\pred(d,n)\})$ time on queries to $\rep$ and the connectivity query between $x_3$ and $x_3'$.
\end{proof}

\subsection{Reduction to Triconnected Graphs in the Fully Dynamic Model}\label{sect:reduction-to-tri-fully-dyn}

We prove Theorem~\ref{thm:reduction-to-tri-fully-dynamic} assuming, for notational simplicity, that the graph is biconnected and associated with an $SPQR$-tree $\mathcal{T}_3$.  It is easy to extend this algorithm to all graphs.

An obvious strategy is to maintain the structure from
Section~\ref{sect:bi-to-tri}, that is, to maintain the ability to evaluate $\rep^*$ on all vertices, $\rep$ on vertices in colored $\mathcal{T}_3$-nodes, and $\rep^{**}$ on the poles of colored nodes.
The first difficulty is that a vertex failure $z$ at some $\phi(z)$ deep in $\mathcal{T}_3$ can cause virtual edge failures in the model graphs of {\em every} ancestor of $\phi(z)$ and potentially change the colored/uncolored status of such nodes.  Thus, we cannot afford to explicitly maintain connectivity oracles at each colored node.\footnote{Even the isolated problem of deciding {\em whether} a node is colored is non-trivial.  In the absence of $R$-nodes this is equivalent to dynamic AND-OR tree evaluation.}
If we do not maintain explicit connectivity oracles then evaluating $\rep^{**}(z)$ seems impossible since it is a function of $O(1)$  $\rep$-queries at {\em each} colored ancestor of $\phi(z)$.

Our solution involves three main ideas.  The first is to partition $\mathcal{T}_3$ into heavy paths, a standard technique that ensures every node-to-root path intersects at most $\log |\mathcal{T}_3| \le \log n$ distinct paths.  This introduces a $\log n$-factor overhead and reduces the problem to one on paths, yet it does not help us stop changes deep in $\mathcal{T}_3$ from propagating upwards.
Two consecutive nodes $u$ and $p(u)$ on a heavy path share two vertices, namely $\pi(u)$.  Our second idea is to make $p(u)$ indifferent to {\em all} changes in $V(G_u) \backslash \pi(u)$ by having it maintain two candidate connectivity oracles, one where $\pi(u)$ is invalid (its endpoints are disconnected in $G_u^-\backslash F$) and one where $\pi(u)$ is valid.  This solves the propagation problem, but we must be able to determine which of $p(u)$'s candidate oracles reflects the {\em actual} failure set $F$.  The third idea is to reduce a $\rep^{**}$-query to $O(\log n)$ 1D product queries over the domain of $2\times 2$ boolean matrices.  Updates and queries to such a structure will take $O(\log n/\log\log n)$ time.

\paragraph{Heavy Path Decomposition}
Let each nonleaf $u\in \mathcal{T}_3$ choose a child $c(u)$ maximizing the number of descendants of $c(u)$.
These choices partition the nodes of $\mathcal{T}_3$ into a set $\mathcal{P}$ of {\em heavy paths} such that all leaf-to-root paths intersect at most $\log n$ such paths.
Let $P(u)$ be the heavy path containing $u$ and $\rt(u) = \rt(P(u))$ be the most ancestral node in $P(u)$.
Define $\tilde{\pi}(u) = \pi(c(u))$ to be the {\em subpolar edge} of $u$
Let $\tilde{G}_u^-$ be the subgraph of $G_u^-$ induced by $V(G_u) \backslash \phi^{-1}(\mathcal{T}_3(c(u)))$,
that is, $G_u^-$ excluding vertices in $V(G_{c(u)}) \backslash \pi(c(u))$.

\paragraph{Twin Connectivity Oracles}
Each non-leaf $u\in \mathcal{T}_3$ maintains two connectivity oracles for $G_u^-$ based on the current set $F$ of failed vertices.
Let $F(u) = F \cap \phi^{-1}(\mathcal{T}_3(u))$ be the failed vertices in $V(G_u)$, excluding the poles $\pi(u)$, if $u$ is not the root of $\mathcal{T}_3$.
Let $F^0(u) = F(u) \cup\phi^{-1}(\mathcal{T}_3(c(u)))$ 
and $F^1(u) = F(u) \backslash \phi^{-1}(\mathcal{T}_3(c(u)))$ reflect versions of reality where all or none of the vertices in $c(u)$'s subtree fail.
The oracle $\mathcal{A}_u^1$ answers connectivity queries between $V(M_u)$ vertices in the graph $G_u^- \backslash F^1(u)$
whereas $\mathcal{A}_u^0$ is defined with respect to $F^0(u)$.
If we restrict our attention to connectivity in $V(M_u)$, at least one of $\mathcal{A}_u^0$ and $\mathcal{A}_u^1$ behaves correctly,
with respect to the failure set $F(u)$.
In order to determine which one is correct let us dispense with the old blue/purple/uncolored coloring system\footnote{Note that it is problematic to assign $\mathcal{T}_3$-nodes these colors since $u$ may be colored purple according to $\mathcal{A}_u^0$ but not $\mathcal{A}_u^1$.} from Section~\ref{sect:bi-to-tri} and designate {\em all} $\mathcal{T}_3$-nodes {\em white}, {\em grey}, or {\em black}.  A node $u$ is white if the poles $\pi(u)$ are disconnected according to $\mathcal{A}_u^1$, black if they are connected according to $\mathcal{A}_u^0$, 
and gray if they are connected according to $\mathcal{A}_u^1$ but not $\mathcal{A}_u^0$.  In other words, in a gray $u$, all paths between $\pi(u)$ in $G_u^-\backslash F(u)$ go through the subpoles $\tilde{\pi}(u)$.  If $u$ is a leaf it has no subpoles and therefore cannot be gray, since $\mathcal{A}_u^0 = \mathcal{A}_u^1$.  Before continuing, let us note that before any vertex failures occur, all non-leaf $S$-nodes are gray (removing $\pi(u)$ and $\tilde{\pi}(u)$ disconnects the poles in $M_u$, which is a cycle) and all $P$- and $R$-nodes are black since there are at least three edge-disjoint paths between $\pi(u)$ in $M_u$.

In order to determine whether the poles $\pi(u)$ are actually connected in $G_u^-\backslash F(u)$ it suffices to maintain a predecessor structure
on all non-gray nodes in $P(u)$, which we view as oriented towards $\rt(P(u))$.  
Let $u'$ be the predecessor of $u$ on $P(u)$, that is, the first non-gray descendant of $u$ on $P(u)$, which may be $u$ itself. (There must be such a $u'$ since leaves cannot be gray.)
If $u'$ is black then the poles $\pi(u)$ are connected in $G_u^-\backslash F(u)$, and if white, they are not.   
It takes $O(\log\log n)$ time per query or to delete/insert any node into the predecessor structure, which corresponds to a color transition to/from gray.

\begin{Lem}
A failure or recovery can invalidate the color designation of at most $\log n$ $\mathcal{T}_3$-nodes.
\end{Lem}

\begin{proof}
Consider the failure or recovery of a vertex $z$ with $\phi(z) = u$ and $P=P(u)$.
This change might cause the color of $u$ to change, but, by definition, cannot invalidate the color
of any other $u'\in P$.  Now consider $v = p(\rt(P))$, if $\rt(P)$ is not already the root of $\mathcal{T}_3$.
The twin connectivity oracles for $v$ entertain the possibility that $\tilde{\pi}(v) \in E(M_v)$ is invalid or valid,
but both must represent the {\em actual} validity of the edge $\pi(\rt(P)) \in E(M_v)$, which may have changed 
due to the failure/recovery of $z$.  Thus, changing the status of $z$ may affect $v$ but no other nodes on $P(v)$ and, in general,
affects up to one node on each heavy path traversed from $\phi(z)$ to the root of $\mathcal{T}_3$.
\end{proof}

\paragraph{Dynamic Matrix Products}

Suppose $P = (u_1,\ldots,u_{|P|})$ is a heavy path.  If we are only concerned with connectivity between poles of nodes in $P$,
each node $u_i$ is effectively in one of 16 states depending on whether each pole in $\tilde{\pi}(u_i)$ is connected, in the graph $\tilde{G}_{u_i}^- \backslash F(u_i)$, to each pole in $\pi(u_i)$.\footnote{Since connectivity is transitive, not all 16 states are possible.}  Let $\sigma_i$ be a $2\times 2$ boolean matrix representing this connectivity.  (In other words, a node's poles are ordered in some consistent fashion to map to rows and columns.)  It follows that the boolean product $\sigma_i\cdots\sigma_j$ represents the connectivity between $\tilde{\pi}(u_i)$ and $\pi(u_j)$ in the graph $(\tilde{G}_{u_i}^- \cup \cdots \cup \tilde{G}_{u_j}^-) \backslash F(u_j)$.  
Our algorithm for computing $\rep^{**}(z)$ (and answering connectivity queries) 
requires a dynamic data structure for answering various product queries.

\begin{enumerate}
\item \texttt{init}$()$: Set $\sigma_i \leftarrow \scalebox{.5}{$\left(\begin{array}{cc}1&1\\1&1\end{array}\right)$}$, for $1< i \le |P|$.
\item \texttt{update}$(i,\sigma)$: $\sigma_i \leftarrow \sigma$, where $\sigma$ is a $2\times 2$ boolean matrix.
\item \texttt{product}$(i,j)$: Return $\sigma_i\cdots \sigma_j$.
\item \texttt{search}$(x,x',i)$: Return $\max\{j\ge i-1 \::\: x(\sigma_i\wedge x')\cdots (\sigma_j \wedge x') \neq 0\}$,
given boolean row vector $x$ and $2\times 2$ mask $x'$.
\end{enumerate}

\begin{Lem}
After $o(n)$ preprocessing, all operations (\texttt{update}, \texttt{product}, \texttt{search}) can be executed
in $O(\log |P|/\log\log n) = O(\log n/\log\log n)$ time.
\end{Lem}

\begin{proof} (sketch)
Each $\sigma_i$ requires 4 bits to represent, so $\Theta(\log n)$ matrices can be packed into one machine word.  
We maintain a $\Theta(\log n)$-way tree over the array, where $\{\sigma_i\}_{1<i\le |P|}$ are at the leaves and internal nodes store the product of their descendant leaves.  All operations either need to update or retrieve $O(1)$ words of information at each of $O(\log |P|/\log\log n)$ levels, which can be done in $O(\log |P|/\log\log n)$ time by tabulating various functions on $\Theta(\log n)$ bits in advance, in $o(n)$ time.
\end{proof}

Define $\rep^{**}_P(z)$ for $\phi(z) \in P$, to be the pole most ancestral to $\phi(z)$ connected to $z$ in $G_{\rt(P)}^- \backslash F$, or $z$ if there is no such pole.
If $\rep^{**}_P(z)$ is not a pole of $\rt(P)$ then clearly $\rep^{**}_P(z) = \rep^{**}(z)$.  In general a $\rep^{**}(z)$ query is easily reduced to $\log n$ $\rep^{**}_P$ queries, for the heavy paths $P$ intersecting the path from $\phi(z)$ to the root of $\mathcal{T}_3$, so we shall focus our attention solely on evaluating $\rep^{**}_P(z)$, where $\phi(z) = u_i \in P$.
We compute $\rep^{**}_P(z)$ as follows.

\begin{enumerate}
\item Our first task is to determine the boolean vector $x$ representing the connectivity between $z$ and the poles $\pi(u_i)$ in $\tilde{G}_{u_i}^- \backslash F$.  Using one predecessor query on the non-gray nodes in $P$ we determine which $\mathcal{A}_{u_i}^\alpha$ reflects reality, where $\alpha\in\{0,1\}$, then determine $x$ with two connectivity queries to $\mathcal{A}_{u_i}^\alpha$.  If $x = 0$ then let $\rep^{**}_P(z) = z$ and halt.

\item Otherwise let $j = $ \texttt{search}$(x,\scalebox{.5}{$\left(\begin{array}{cc}1&1\\1&1\end{array}\right)$},i+1)$.
Let $\rho \subseteq \pi(u_j)$ be the poles reachable from $z$ in $G_{u_j}^- \backslash F(u_j)$.
It follows that $z$ is disconnected from both poles of $u_{j+1}$ in
$G_{u_{j+1}}^- \backslash F(u_{j+1})$, if $u_j \neq \rt(P)$.  If there is some pole in $\rho \backslash F$ let $\rep^{**}_P(z)$ be any such pole and halt.  (Recall that $F(u_j)$ is the set of failures in $V(G_{u_j})$ {\em excluding} its poles.)

\item If the procedure has not halted then $\rho\subseteq F$.
For any index $1\le l\le |P|$ let $\pi(u_l) = (\pi_0^l, \pi_1^l)$.
Without loss of generality suppose $\pi_0^j\in \rho$, and that $j'\in (i,j]$ is the minimum index
such that $\pi_0^{j'} = \cdots = \pi_0^j$ and, if $\pi_1^{j}\in \rho$, $\pi_1^{j'} \neq \pi_1^j$.
(Recall that a pole can appear in an unbounded number of nodes along $P$.)
At least one of the poles in $\tilde{\pi}(u_{j'})$ is connected to $z$.
If $\pi_1^{j'}$ is disconnected from $z$ then $\rep^{**}_P(z)$ can be either pole of $\tilde{\pi}(u_{j'})$ connected to $z$.
Otherwise $\rep^{**}_P(z)$ is the most ancestral pole among $\pi_1^{j'},\ldots,\pi_1^{j-1}$ connected to $z$,
which can be found as follows.  
Let 
$j'' = $ \texttt{search}$(\scalebox{.5}{$\left(\begin{array}{cc}0\\1\end{array}\right)$},
\scalebox{.5}{$\left(\begin{array}{cc}0 & 0\\0 & 1\end{array}\right)$},j'+1)$, where the 
first and second arguments encode that we are only interested in connectivity between $\pi_1^{j'},\ldots,\pi_1^{j-1}$.
If $\pi_1^{j''} \not\in F$ then $\rep^{**}_P(z) = \pi_1^{j''}$, otherwise $\rep^{**}_P(z) = \pi_1^{j'''}$ where $j''' < j''$ is the maximum index such that $\pi_1^{j'''} \neq \pi_1^{j''}$.
\end{enumerate}

The time bounds claimed in Theorem~\ref{thm:reduction-to-tri-fully-dynamic} follow easily.  It takes linear time to build the $SPQR$-tree, decompose it into heavy paths, build the white/black predecessor structure and boolean product structure on each heavy path, and augment it with various pointers, e.g., pointers from each node to the root of its heavy path.  A vertex update on $z$ induces vertex updates in the oracles $\mathcal{A}_{\phi(z)}^0$ and $\mathcal{A}_{\phi(z)}^1$, updates to the color and $2\times 2$ matrix of $\phi(z)$, which take $O(\log n/\log\log n)$ time, 
then $O(\log n)$ edge updates in ancestors of $\phi(z)$ together with $O(\log n)$ color and matrix updates.
A query on $x,x'$ amounts to computing a connectivity query between $\rep^{**}(x)$ and $\rep^{**}(x')$.
A $\rep^{**}$ evaluation reduces to $\log n$ $\rep^{**}_P$ evaluations, which involve a constant number 
of queries and operations on the matrix product data structure, for a total time of $O(\log n(Q_3 + \log n/\log\log n))$.

\bibliographystyle{plain}


\end{document}